\documentclass{article}
\usepackage{amsfonts,amssymb,amsmath,a4wide,paralist}
\usepackage{epsfig,rotating,framed,enumerate,listings}
\usepackage{color}
\usepackage{paralist}
\usepackage{marginnote,nicefrac}
\usepackage{authblk}
\usepackage{tikz}
\usepackage{xfrac}

\usepackage[paper=a4paper,left=31mm,right=31mm,top=30mm, bottom=25mm]{geometry}

\usepackage{hyperref}

\usepackage{algorithmic}
\usepackage[lined,boxed,ruled,commentsnumbered]{algorithm2e}
\newcommand{\w} {\mbox{W}}

\newcommand{\un} {{\it und}}

\newcommand{\OCN} {\text{OCN}}

\newcommand{\free} {\text{Free}}

\newcommand{\OP} {\text{OP}}
\newcommand{\OCI} {\text{OCI}}

\newtheorem{theorem}{Theorem}[section]
\newtheorem{example}[theorem]{Example}
\newtheorem{lemma}[theorem]{Lemma}
\newtheorem{definition}[theorem]{Definition}
\newtheorem{proposition}[theorem]{Proposition}
\newtheorem{observation}[theorem]{Observation}
\newtheorem{corollary}[theorem]{Corollary}

\newenvironment{proof}{\noindent{\bf Proof~}}{\null\hfill $\Box$\par\medskip}

\definecolor{light-gray}{gray}{0.5}

\begin{document}

\title{Oriented coloring on recursively defined digraphs}

\author[1]{Frank Gurski}
\author[1]{Dominique Komander}
\author[1]{Carolin Rehs}

\affil[1]{\small University of  D\"usseldorf,
Institute of Computer Science, Algorithmics for Hard Problems Group,\newline 
40225 D\"usseldorf, Germany}

\maketitle

\begin{abstract}
Coloring is one of the most famous problems in graph theory. The coloring
problem on undirected graphs has been well studied, whereas there are very
few results for coloring problems on directed graphs.
An oriented $k$-coloring of an oriented graph $G=(V,A)$ is a partition of the
vertex set $V$ into $k$ independent sets such that all the arcs linking
two of these subsets have the same direction. The oriented chromatic
number of an oriented graph $G$ is the  smallest $k$ such that $G$
allows an oriented $k$-coloring.
Deciding whether an acyclic digraph allows an oriented $4$-coloring
is NP-hard. It follows that finding the chromatic number of an oriented
graph is an NP-hard problem, too. This~motivates to consider the problem
on  oriented co-graphs. After giving several characterizations for this graph class, we show
a linear
time algorithm which computes an optimal oriented coloring for an oriented
co-graph. We further prove how the oriented chromatic number can be computed
for the disjoint union and order composition from the oriented chromatic number
of the involved oriented co-graphs. It turns out that within oriented co-graphs the
oriented chromatic number is equal to the length of a longest oriented path plus one.
We also show that the
graph isomorphism problem on oriented co-graphs can be solved in linear
time.

\end{abstract}

\section{Introduction} \label{sec-intro}

Graph coloring is one of the basic problems in graph theory, which has already been considered
in the 19th century. A $k$-coloring for an undirected graph $G$ is a $k$-labeling of the vertices
of $G$ such that no two adjacent vertices have the same label. The smallest $k$ such that a graph
$G$ has a $k$-coloring is named the chromatic number of $G$.
As even the problem whether a graph has a $3$-coloring, is NP-complete, finding the chromatic
number of an undirected graph is an NP-hard problem. However, there are many efficient solutions
for the coloring
problem on special graph classes, like chordal graphs \cite{Gol80},  comparability graphs \cite{Hoa94},
and co-graphs \cite{CLS81}.

Oriented coloring has been introduced much later by Courcelle \cite{Cou94}. One could easily apply the definition
of graph coloring to directed graphs, but as this would not take the direction of the arcs into
account, this would not be very interesting. For such a definition, the coloring of a directed graph
would be the coloring of the underlying undirected graph.

Oriented coloring also considers  the direction of the arcs. An oriented $k$-coloring of an
oriented graph $G=(V,A)$ is a partition of the vertex set $V$ into $k$ independent sets, such
that all the arcs linking two of these subsets have the same direction. In the oriented chromatic
number problem ($\OCN$ for short) there is given some oriented graph $G$ and some integer $c$ and 
one has to
decide whether there is an  oriented $c$-coloring for $G$. Even the restricted
problem, when $c$ is constant and does not belong to the input  ($\OCN_{c}$ for short), is hard.
$\OCN_{4}$ is NP-complete even for DAGs \cite{CD06}, whereas the undirected problem is
easy for trees.

Right now, the definition of oriented coloring is mostly considered for undirected graphs.
There~the maximum value $\chi_o(G')$ of all possible orientations $G'$ of an undirected graph
$G$ is considered.
For~several special undirected graph classes the oriented chromatic number has been bounded.
Among~these are outerplanar graphs \cite{Sop97}, planar graphs \cite{Mar13}, and
Halin graphs \cite{DS14}. In \cite{Gan09},
Ganian has shown  an FPT-algorithm for $\OCN$ w.r.t.\ the parameter
tree-width (of the underlying undirected graph).
Further, he has shown that $\OCN$ is DET-hard (DET is the class of decision
problems which are reducible in logarithmic space to 
the problem of computing the determinant of an integer valued $n\times n$-matrix.) for classes of
oriented graphs, such that the underlying undirected class has bounded rank-width.

\begin{figure}[ht]
\begin{center}
 {\includegraphics[width=1.7cm]{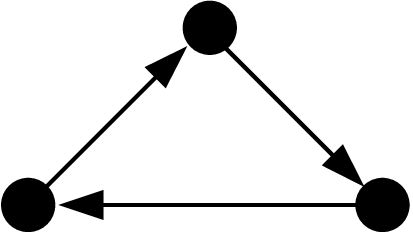}} ~~~~~~~~~ {\includegraphics[width=1.7cm]{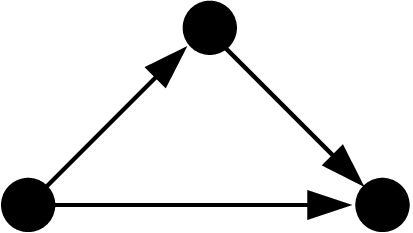}} 
\end{center}
\caption{Special oriented graphs: oriented cycle  $\protect\overrightarrow{C_3}$ and transitive tournament $\protect\overrightarrow{T_3}$.}\label{ex-spe}
\end{figure}

Oriented coloring of special digraph classes seems not to be investigated up to now.
The main reason is that the oriented chromatic number of the disjoint union of two oriented graphs can
be larger than the maximum oriented chromatic number of the involved graphs
(cf.\ Figure \ref{ex-spe} and Example~\ref{ex-disj-u2}). In this paper, we
consider the oriented coloring problem restricted to oriented co-graphs, which
are obtained from directed co-graphs \cite{BGR97} by omitting the series operation.
Oriented co-graphs were already analyzed by Lawler
in \cite{Law76} and \cite{CLS81}~(Section 5) using the notation 
of transitive series parallel (TSP) digraphs.
We give several characterizations for oriented co-graphs and show that for oriented co-graphs, the
oriented chromatic number of the disjoint union of oriented graphs
is equal to the maximum oriented chromatic number of the involved graphs. Further,
we show that for every oriented graph the oriented chromatic number of the order composition
of oriented graphs is equal to the sum of the oriented chromatic numbers of the involved graphs.
To show this, we introduce an algorithm that computes an optimal oriented coloring and thus, the
oriented chromatic number of oriented co-graphs in linear time.
We also consider the longest oriented path problem on oriented co-graphs.
It turns out that within oriented co-graphs the
oriented chromatic number is equal to the length of a longest oriented path plus one.
Further, we give a linear time  algorithm for the graph isomorphism problem on oriented co-graphs.
Since oriented co-graphs have a directed NLC-width of one \cite{GWY16},
our results provide a useful basis for exploring the complexity of $\OCN$ related to width parameters
(cf.~Section~\ref{sec-con}).

\section{Preliminaries}\label{intro}

\subsection{Graphs and Digraphs}

We use the notations of Bang-Jensen and Gutin \cite{BG09} for graphs and digraphs.

For some given digraph $G=(V,A)$, we define
its underlying undirected graph by ignoring the directions of the edges, i.e.
$\un(G)=(V,\{\{u,v\}~|~(u,v)\in A, u,v\in V\})$
and for some class of digraphs $X$, let
$\un(X)=\{\un(G)~|~ G\in X\}$.
For some (di)graph class $F$ we define
$\free(F)$ as the set of all (di)graphs $G$, such that no induced sub(di)graph 
of $G$ is isomorphic to a member of  $F$.

An {\em oriented graph} is a digraph with no loops and no opposite arcs.
We recall some special oriented graphs.
By $$\overrightarrow{P_n}=(\{v_1,\ldots,v_n\},\{ (v_1,v_2),\ldots, (v_{n-1},v_n)\}),$$ $n \ge 2$,
we denote the oriented path on $n$ vertices,
by $$\overrightarrow{C_n}=(\{v_1,\ldots,v_n\},\{(v_1,v_2),\ldots, (v_{n-1},v_n),(v_n,v_1)\}),$$ $n \ge 3$,
we denote the oriented cycle on $n$ vertices and by
$\overrightarrow{T_n}$ we denote the transitive tournament on $n$~vertices.

\subsection{Undirected  Co-Graphs}

Let $G_1=(V_1,E_1), \ldots, G_k=(V_k,E_k)$ be $k$ vertex-disjoint graphs.
\begin{itemize}
\item
The {\em disjoint union} of $G_1, \ldots, G_k$,
denoted by $G_1 \cup \ldots \cup G_k$,
is the graph with vertex set $V_1\cup \ldots \cup V_k$ and
edge set $E_1\cup \ldots \cup E_k$.

\item
The {\em join composition} of $G_1, \ldots, G_k$,
denoted by $G_1\times \ldots \times G_k$,
is defined by their disjoint union plus all possible edges between
vertices of $G_i$ and $G_j$  for all $1\leq i,j\leq k$, $i\neq j$.
\end{itemize}

The set of all graphs, which can be defined from
a single vertex graph by applying the disjoint union  and join composition,
is characterized as the set of all co-graphs. It is well known that
co-graphs are precisely the $P_4$-free graphs \cite{CLS81}.

\subsection{Undirected Graph Coloring}

\begin{definition}[Graph Coloring]
A {$k$-coloring} of a graph $G=(V,E)$  is a mapping $c:V\to \{1,\ldots,k\}$
such~that:
\begin{itemize}
	\item $c(u)\neq c(v)$ for every $\{u,v\}\in E$
\end{itemize}
The  {chromatic number} of $G$, denoted by $\chi(G)$, is the smallest $k$
such that $G$ has a $k$-coloring.
\end{definition}

On undirected co-graphs, the graph coloring problem is easy to solve
by the following result proven by Corneil et al.:

\begin{lemma}[\cite{CLS81}]\label{colo-und}
Let $G_1, \ldots, G_k$ be $k$ vertex-disjoint graphs.
\begin{enumerate}
\item $\chi(G_1\cup \ldots \cup G_k)=\max(\chi(G_1),\ldots,\chi(G_k))$

\item $\chi(G_1\times \ldots \times  G_k)= \chi(G_1) +\ldots +\chi(G_k)$
\end{enumerate}
\end{lemma}

\begin{proposition} Let $G$ be a co-graph. Then,
$\chi(G)$  can be computed in linear time.
\end{proposition}

The undirected coloring problem, i.e., computing $\chi(G)$, can be solved by
an FPT-algorithm w.r.t.\ the  tree-width of the input graph \cite{Gur08c}.
In contrast, this is not true for clique-width, since
it has been shown  in \cite{FGLS09}, that the undirected coloring problem
is $\w[1]$-hard w.r.t.\ the  clique-width of the input graph.
That is, under reasonable assumptions an XP-algorithm is the best one can hope for. Such~algorithms
are known, see \cite{EGW01a}.

\subsection{Directed Co-Graphs}

The following operations for digraphs
have already been considered  by Bechet et al.\ in \cite{BGR97}.
Let~$G_1=(V_1,E_1), \ldots, G_k=(V_k,E_k)$ be $k$ vertex-disjoint digraphs.
\begin{itemize}

\item
The {\em disjoint union} of $G_1, \ldots, G_k$,
denoted by $G_1 \oplus \ldots \oplus G_k$,
is the digraph with vertex set $V_1\cup \ldots \cup V_k$ and
arc set $E_1\cup \ldots \cup E_k$.

\item
The {\em series composition} of $G_1,\ldots, G_k$,
denoted by $G_1\otimes \ldots \otimes G_k$,
is defined by their disjoint union plus all possible arcs between
vertices of $G_i$ and $G_j$ for all $1\leq i,j\leq k$, $i\neq j$.

\item
The {\em order composition} of $G_1, \ldots, G_k$,
denoted by $G_1\oslash \ldots \oslash G_k$,
is defined by their disjoint union plus all possible arcs from
vertices of $G_i$ to vertices of $G_j$ for all $1\leq i < j\leq k$.
\end{itemize}

The set of all
digraphs which can be defined by the disjoint union, series composition, and order composition
is characterized as the set of all directed co-graphs \cite{BGR97}.
Obviously, for every directed co-graph we can define a tree structure,
denoted as the {\em di-co-tree}. The leaves of the di-co-tree represent the
vertices of the graph and the inner nodes of the di-co-tree  correspond
to the operations applied on the subexpressions defined by the subtrees.
For every directed co-graph one can construct a  di-co-tree in linear time,
see \cite{CP06}.

In \cite{BM14} it is shown that the weak $k$-linkage problem can be solved
in polynomial time for directed co-graphs. By the recursive structure
there exist dynamic programming algorithms
to compute
the size of a largest edgeless subdigraph,
the~size of a largest subdigraph which is a tournament,
the size of a largest semicomplete subdigraph, and
the size of a largest complete subdigraph
for every directed co-graph  in linear time. Also
the Hamiltonian path,  Hamiltonian cycle, regular
subdigraph, and directed cut problem are polynomial
on directed co-graphs \cite{Gur17a}.
Calculs of directed co-graphs were also considered in connection
with pomset logic in \cite{Ret98}.
Further, the directed path-width and  directed tree-width
can be computed in linear time for  directed co-graphs \cite{GR18c}.

In \cite{CP06}, it has been shown
that directed co-graphs can be characterized by the eight forbidden induced
subdigraphs shown in Figure  \ref{F-co}.

\begin{figure}
\begin{center}
\begin{tabular}{cccccccc}
{\includegraphics[width=1.8cm]{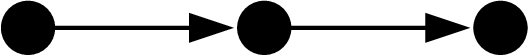}} &~~~&{\includegraphics[width=1.8cm]{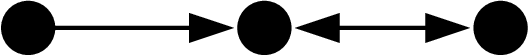}}&~~~&{\includegraphics[width=1.8cm]{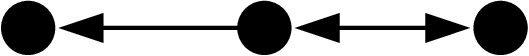}}&~~~&{\includegraphics[width=1.5cm]{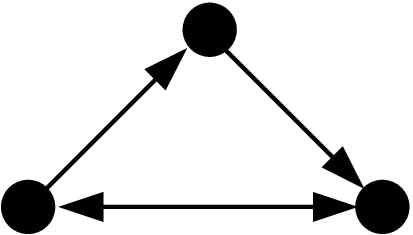}} &\\
 $D_1$   &  &    $D_2$   &  &   $D_3$   &  &   $D_4$   &  \\ 
&&&&&&&\\
{\includegraphics[width=1.5cm]{g5.eps}} &~~~&{\includegraphics[width=1.3cm]{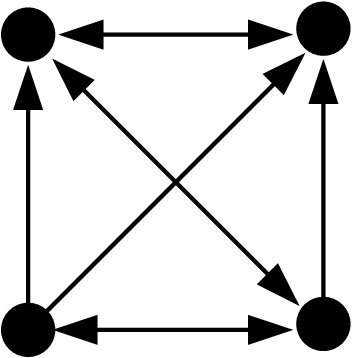}}&&{\includegraphics[width=2.7cm]{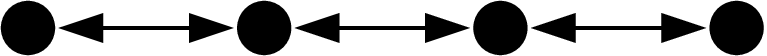}}&&{\includegraphics[width=2.7cm]{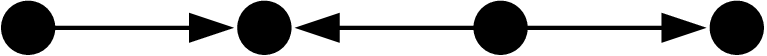}}&\\
 $D_5$   &  &    $D_6$   &  &   $D_7$   &  &   $D_8$   &  \\ 
\end{tabular}
\end{center}
\caption{The eight forbidden induced subdigraphs for directed co-graphs.}
\label{F-co}
\end{figure}

\section{Oriented Co-Graphs}\label{dirco}

Oriented colorings are defined on oriented graphs, which are digraphs with no
bidirected edges. Therefore we introduce oriented co-graphs by omitting 
the series operation from the definition of directed co-graphs, as given
in  \cite{BGR97}.

\begin{definition}[Oriented Co-Graphs]\label{ocog}
The class of { oriented co-graphs} is recursively defined as follows.
\begin{enumerate}
\item Every digraph on a single vertex $(\{v\},\emptyset)$,
denoted by $\bullet$, is an { oriented co-graph}.
\item If
$G_1, \ldots, G_k$ are $k$ vertex-disjoint
oriented co-graphs, then
\begin{enumerate}\label{op-oc}
\item
$G_1\oplus \ldots \oplus G_k$ and
\item
$G_1\oslash \ldots \oslash G_k$  are { oriented co-graphs}.
\end{enumerate}
\end{enumerate}
\end{definition}

The class of oriented co-graphs was already analyzed by Lawler
in \cite{Law76} and \cite{CLS81}~(Section 5) using the notation
of {\em transitive series parallel (TSP) digraphs}.
A digraph $G=(V,A)$ is called {\em transitive}, if for
every pair $(u,v)\in A$ and $(v,w)\in A$ of arcs
with $u\neq w$ the arc $(u,w)$ also belongs to $A$.

\begin{theorem}[\cite{CLS81}]
A graph $G$ is a co-graph if and only if there exists an orientation $G'$ of  $G$,
such that $G'$ is an oriented co-graph.
\end{theorem}

A di-co-tree $T$ is {\em canonical} if
on every path from the root to the leaves of $T$, the labels
disjoint union and order operation strictly alternate.
Since the disjoint union $\oplus$ and the order composition $\oslash$ are associative,
we always can assume canonical di-co-trees.

\begin{lemma}\label{lem-level}
Let $G$ be an oriented co-graph and $T$ be
a di-co-tree for $G$. Then, $T$ can be transformed in linear time
into a canonical di-co-tree for $G$.
\end{lemma}

The recursive definitions of  oriented and undirected co-graphs lead to the following
observation.

\begin{observation}\label{obs-oco1}
For every oriented co-graph $G$ the
underlying undirected graph $\un(G)$ is a co-graph.
\end{observation}

The reverse direction of this observation only holds under certain conditions,
see Theorem \ref{th-ch-oco}. By~$\overleftrightarrow{P_2}=(\{v_1,v_2\},\{(v_1,v_2),(v_{2},v_1)\})$
we denote the complete biorientation of a path on two vertices.

\begin{lemma}\label{le-co-t}
Let $G$ be a digraph, such that $G\in\free(\{\overleftrightarrow{P_2},D_1,D_5\})$. Then,
it holds that $G$ is transitive.
\end{lemma}

\begin{proof}
Let  $(u,v),(v,w)\in A$ be two arcs of  $G=(V,A)$.
Since $G\in\free(\{\overleftrightarrow{P_2}\})$, we know that $(v,u),(w,v)\not\in A$.
Further, since  $G\in\free(\{D_1,D_5\})$, we know that $u$ and $w$ are
connected either only by
$(u,w)\in A$ or by $(u,w)\in A$ and $(w,u)\in A$, which
implies that $G$ is transitive.
\end{proof}

Oriented co-graphs can be characterized by forbidden subdigraphs as follows.

\begin{theorem}\label{th-ch-oco}
Let $G$ be a digraph. The following properties are equivalent:
\begin{enumerate}
	\item\label{ch-oc1} $G$ is an oriented co-graph.
	\item\label{ch-oc2} $G\in \free(\{D_1, D_5, D_8, \overleftrightarrow{P_2}\})$.
        \item\label{ch-oc3a} $G\in \free(\{D_1, D_5, \overleftrightarrow{P_2}\})$
		   and $\un(G)\in\free(\{P_4\})$.
        \item\label{ch-oc3} $G\in \free(\{D_1, D_5, \overleftrightarrow{P_2}\})$
		   and $\un(G)$ is a co-graph.

	\item\label{ch-oc4}  $G$ has directed NLC-width $1$ and $G\in \free(\{\overleftrightarrow{P_2}\})$.
	\item\label{ch-oc5}  $G$ has directed clique-width at most $2$ and $G\in \free(\{\overleftrightarrow{P_2}\})$.
        \item\label{ch-oc6} $G$ is  transitive and $G\in \free(\{\overleftrightarrow{P_2}, D_8\})$.
\end{enumerate}
\end{theorem}

\begin{proof}
$(\ref{ch-oc1})\Rightarrow(\ref{ch-oc2})$ If $G$ is an oriented co-graph, then
$G$ is a directed co-graph and by \cite{CP06} it holds that $G\in \free(\{D_1,\ldots, D_8\})$. Furthermore, $G\in \free(\{\overleftrightarrow{P_2}\})$
because of the missing series composition. This leads to $G\in \free(\{D_1, D_5, D_8, \overleftrightarrow{P_2}\})$.
$(\ref{ch-oc2})\Rightarrow(\ref{ch-oc1})$
If $G\in \free(\{D_1, D_5, D_8, \overleftrightarrow{P_2}\})$, then $G\in \free(\{D_1,\ldots, D_8\})$ and is $G$ a directed co-graph.
Since $G\in \free(\{\overleftrightarrow{P_2}\})$,
there is no series operation in any
construction of $G$ which implies that $G$ is an oriented co-graph.
$(\ref{ch-oc3a})\Leftrightarrow(\ref{ch-oc3})$ Since co-graphs are precisely the $P_4$-free graphs \cite{CLS81}.
$(\ref{ch-oc2})\Rightarrow(\ref{ch-oc6})$ By Lemma \ref{le-co-t}.
$(\ref{ch-oc6})\Rightarrow(\ref{ch-oc2})$ If $G$ is transitive, then
$G\in\free(\{D_1, D_5\})$.
$(\ref{ch-oc1})\Leftrightarrow(\ref{ch-oc4})$ and  $(\ref{ch-oc1})\Leftrightarrow(\ref{ch-oc5})$ By \cite{GWY16}.
$(\ref{ch-oc1})  \&(\ref{ch-oc2})  \Rightarrow(\ref{ch-oc3})$ By Observation \ref{obs-oco1}. $(\ref{ch-oc3a})\Rightarrow(\ref{ch-oc2})$ If $\un(G)$
does not contain a $P_4$, then $G$ can not contain any orientation of $P_4$.
\end{proof}

Among others are two subclasses of oriented co-graphs, which will be of interest within
our results.  By restricting within Definition \ref{ocog} (\ref{op-oc}) to $k=2$ and graph $G_1$ or $G_2$
to an edgeless graph or to a single vertex, we obtain the class of all
{\em oriented simple co-graphs} or {\em oriented threshold graphs}, respectively.
The class of oriented threshold graphs has been introduced by Boeckner in \cite{Boe18}.

\section{Graph Coloring on Recursively Defined Digraphs}

\subsection{Oriented Graph Coloring Problem}

Oriented graph coloring has been introduced by Courcelle \cite{Cou94} in 1994. Most results
on this problem consider orientations of undirected graphs. Now, we consider oriented graph
coloring on recursively defined oriented graph classes.

\begin{definition}[Oriented Graph Coloring \cite{Cou94}]\label{def-oc}
An {oriented $k$-coloring} of an oriented graph $G=(V,A)$ is a mapping $c:V\to \{1,\ldots,k\}$,
such that:
\begin{itemize}
	\item $c(u)\neq c(v)$ for every $(u,v)\in A$
	\item $c(u)\neq c(y)$ for every two arcs $(u,v)\in A$ and $(x,y)\in A$ with $c(v)=c(x)$
\end{itemize}
The { oriented chromatic number} of $G$, denoted by $\chi_o(G)$, is the smallest $k$,
such that $G$ has an oriented $k$-coloring.
The vertex sets $V_i=\{v\in V\mid c(v)=i\}$, $1\leq i\leq k$, divide a partition of
$V$ into so called { color~classes}.
\end{definition}

For two oriented graphs $G_1=(V_1,A_1)$ and $G_2=(V_2,A_2)$ a {\em homomorphism} from $G_1$ to $G_2$, $G_1 \to G_2$ for short,
is a mapping $h: V_1 \to V_2$, such that  $(u,v) \in A_1$ implies  that
$(h(u),h(v)) \in A_2$.  The oriented graphs $G_1$ and $G_2$ are {\em homomorphically equivalent},
if there is a homomorphism from $G_1$ to $G_2$ and  one from $G_2$ to $G_1$.
A homomorphism from $G_1$ to $G_2$ can be regarded as an oriented coloring of $G_1$
that uses the vertices of $G_2$ as colors classes. This leads to equivalent definitions for
oriented coloring and oriented chromatic number.
There is an oriented $k$-coloring of an oriented graph $G_1$
if and only if there is a homomorphism from $G_1$ to some oriented graph $G_2$ on $k$ vertices.
That is, the oriented chromatic number of $G$ is the minimum number of vertices in an
oriented graph $G_2$, such that there is a homomorphism from $G_1$ to $G_2$.
Obviously, $G_2$ can be chosen as a tournament.

\begin{observation}\label{oc-tournament}
There is an oriented $k$-coloring of an oriented graph $G_1$
if and only if there is a homomorphism from $G_1$ to some tournament $G_2$ on $k$ vertices.
Further, the oriented chromatic number of $G$ is the minimum number of vertices in a
tournament $G_2$, such that there is a homomorphism from $G_1$ to $G_2$.
\end{observation}

\begin{lemma}\label{le-isubdigraph}
Let $G$ be an oriented graph and $H$ be a subdigraph
of $G$, then $\chi_o(H)\leq \chi_o(G)$.
\end{lemma}

\begin{example}\label{ex-color}For oriented paths and oriented cycles we know:
$\chi_o(\overrightarrow{P_2})=2$,  $\chi_o(\overrightarrow{P_3})=3$,
$\chi_o(\overrightarrow{C_4})=4$, $\chi_o(\overrightarrow{C_5})=5$.
\end{example}

An oriented graph $G=(V,A)$ is an {\em oriented clique} ({\em o-clique}) if $\chi_o(G) = |V|$.
Thus all graphs given in Example \ref{ex-color} are oriented cliques.

\begin{itemize}
\item[\textbf{Name}] Oriented Chromatic Number ($\OCN$)

\item[\textbf{Instance}]  An oriented graph $G=(V,A)$ and a positive integer $c \leq |V|$.

\item[\textbf{Question}]  Is there an oriented $c$-coloring for $G$?
\end{itemize}

If $c$ is constant and not part of the input, the corresponding problem
is denoted by $\OCN_{c}$. Even~for DAGs $\OCN_{4}$ is NP-complete  \cite{CD06}.

The definition of oriented coloring is also used for undirected graphs. For an
undirected graph $G$ the 
maximum value $\chi_o(G')$ of all possible orientations $G'$ of
$G$ is considered. In this sense, every tree  has oriented chromatic number at most $3$.
For several further graph classes there exist bounds on the oriented number.
Among these are outerplanar graphs \cite{Sop97}, planar graphs \cite{Mar13}, and
Halin graphs \cite{DS14}.

\subsection{Oriented Graph Coloring for Oriented Graphs}

Oriented graph coloring has not yet been considered for recursively defined graphs,
though it has been analyzed for some graph operations. In this section we show 
results of oriented graph coloring on some graph operations and provide algorithms
for recursively defined oriented graph classes. This will also be useful for the
following section.

First, we give some results on the oriented graph coloring for general
recursively defined oriented graphs. These results will be very useful
to prove our results for oriented co-graphs in the next section.

\begin{lemma}\label{le-dec}Let $G_1,\ldots,G_k$ be $k$ vertex-disjoint oriented graphs.
Then the following equations holds:
%
\begin{enumerate}
\item $\chi_o(G_1\oplus \bullet)=\chi_o(G_1)$

\item\label{le-dec-b} $\chi_o(G_1 \oplus \ldots\oplus G_k) \geq \max(\chi_o (G_1), \ldots, \chi_o(G_k))$

\item $\chi_o(G_1\oslash \ldots \oslash G_k)= \chi_o(G_1) + \ldots + \chi_o(G_k)$
\end{enumerate}
\end{lemma}

\begin{proof}~~
\begin{enumerate}
\item
$\chi_o(G_1\oplus \bullet)\leq \chi_o(G_1)$

Since no new arcs are inserted $G_1$ can keep its colors. The added
isolated vertex gets a color of $G_1$ in order to obtain a valid coloring
for $G_1\oplus \bullet$.

$\chi_o(G_1\oplus \bullet)\geq \chi_o(G_1)$

This relation
holds by Lemma \ref{le-isubdigraph},
since $G_1$ is an induced subdigraph of $G_1 \oplus \bullet$.

\item $\chi_o(G_1 \oplus \ldots\oplus G_k)\geq  \max(\chi_o (G_1), \ldots, \chi_o(G_k))$

Since the digraphs $G_1,\ldots, G_k$ are induced subdigraphs of digraph $G_1 \oplus \ldots\oplus G_k$,
all values $\chi_o (G_1),\ldots, \chi_o(G_k)$ lead to
a lower bound for the number of necessary colors
of the combined graph by Lemma \ref{le-isubdigraph}.

\item
$\chi_o(G_1\oslash  \ldots \oslash  G_k)\leq \chi_o(G_1) + \ldots + \chi_o(G_k)$

For $1\leq i \leq k$ let $G_i=(V_i,A_i)$ and $c_i:V_i\to\{1,\ldots,\chi_o(G_i)\}$
a coloring for $G_i$.
For $G_1\oslash  \ldots \oslash  G_k=(V,A)$ we define a mapping $c:V\to \{1,\ldots,\sum_{j=1}^{k}\chi_o(G_i) \}$ as follows.
\[
c(v)=\left\{ \begin{array}{ll}
c_1(v)                        &  {\rm\ if\ } v \in V_{G_1} \\

c_i(v)+\sum_{j=1}^{i-1}\chi_o(G_i)   &  {\rm\ if\ } v \in V_{G_i},~2\leq i\leq k. \\
\end{array}\right.
\]

The mapping $c$ satisfies the definition of an oriented coloring, because
no two adjacent vertices from $G_i$, $1\leq i \leq k$, have the same
color by assumption and  by definition
of $c$. For $1\leq i \neq j \leq k$
a vertex of $G_i$ and a vertex of $G_j$  are always adjacent, but
never colored equally by  definition of  $c$.

Further, the arcs between two color classes of every $G_i$, $1\leq i \leq k$,
have the same direction by  definition of  $c$. 
For $1\leq i \neq j \leq k$ the arcs between a
color class of $G_i$ and a color class of $G_j$ 
have the same direction by definition of the order operation.

$\chi_o(G_1\oslash  \ldots \oslash  G_k)\geq \chi_o(G_1) + \ldots + \chi_o(G_k)$

Since every $G_i$, $1\leq i \leq k$, is an induced subdigraph of the combined
graph, all values $\chi_o (G_1),\ldots, \chi_o(G_k)$ lead to
a lower bound for the number of necessary colors
of the combined graph by Lemma \ref{le-isubdigraph}. Further,
the order operations implies that for every $1\leq i\neq j\leq k$ no vertex in $G_i$ can be
colored in the same way as a vertex in $G_j$.
Thus, $\chi_o(G_1) + \ldots + \chi_o(G_k)$
leads to a lower bound for the number of necessary colors
of the combined graph.

\end{enumerate}
This shows the statements of the lemma.
\end{proof}

By Lemma \ref{le-dec}, we can solve oriented coloring
for oriented simple co-graphs and thus, also for subclasses, such as oriented
threshold graphs and transitive tournaments, in linear time.

\begin{proposition} Let $G$ be an oriented simple co-graph. Then, it holds that  
$\chi_o(G)=\chi(\un(G))=\omega(\un(G))$ and all values can be
computed in linear time.
\end{proposition}

It is not easy to generalize these results to oriented co-graphs.
To do so, we would need to compute the oriented chromatic number of the
disjoint union of two oriented co-graphs with at least two vertices.
But it is not possible to compute this oriented chromatic number of the
disjoint union of general oriented graphs from the  oriented chromatic numbers
of the involved graphs.
In~Lemma~\ref{le-dec}~(\ref{le-dec-b}) we only show a lower bound.
The following example proves that in general this can not be strengthened to equality.

\begin{example}\label{ex-disj-u2}
The two graphs $\overrightarrow{C_3}$ and $\overrightarrow{T_3}$
in Figure \ref{ex-spe} have the same oriented
chromatic number  $\chi_o(\overrightarrow{C_3})=\chi_o(\overrightarrow{T_3})=3$, but
their disjoint union needs more colors.
\end{example}

On the other hand, there are several examples for which  the disjoint union does not need more
than $\max(\chi_o(G_1),\chi_o(G_2))$ colors, such as  the union of two isomorphic oriented
graphs. By~Theorem~\ref{th-ch-oco}, we know that $\overrightarrow{T_3}$, shown in Figure  \ref{ex-spe}, is an oriented co-graph, but
$\overrightarrow{C_3}$, shown in Figure  \ref{ex-spe}, is not an oriented co-graph.
Consequently, the question arises whether oriented coloring could be closed under disjoint union,
when restricted to oriented co-graphs.

\subsection{Oriented Graph Coloring for Oriented Co-Graphs}

In order to solve $\OCN$ restricted to oriented co-graphs $G$ we created a procedure,
which is shown in Algorithm \ref{fig:algorithm3x}. 
The method
traverses a canonical di-co-tree $T$ for $G$ using a depth-first search,
such that for every inner vertex the children are visited from left to right.
For every inner vertex $u$ of $T$, we store two values ${\rm in}[u]$ and ${\rm out}[u]$.
These values ensure that the vertices of $G$, corresponding to the
leaves of the subtree, rooted at $u$ will we labeled
by labels $\ell$, such that  ${\rm in}[u]\leq \ell \leq {\rm out}[u]$.
For every leaf vertex $u$ of $T$, we additionally store the
label of the corresponding vertex of $G$ in ${\rm color}[u]$.
These~values lead to an optimal oriented coloring of $G$ by the next theorem.

\vspace{12pt}
\begin{algorithm}[H]
{\strut\footnotesize \bf procedure {\sc Label}$(G,u,i)$}
\footnotesize
\smallskip
\begin{tabbing}
xxxx \= xxxx \= xxxx \= xxxx \= \kill
{\bf if} ($u$ is a leaf of $T$)  \{  \\
\> ${\rm color}[u]=i$; ${\rm in}[u]=i$; ${\rm out}[u]=i$;\\
\} \\
{\bf else} \{  \\
\> ${\rm in}[u]=i$; ${\rm out}[u]=0$; \\
\> {\bf for all}   children $v$ of $u$ {\bf from left to right do} \{  \\
\> \> $j=${\sc Label}$(G,v,i)$; \\
\> \> {\bf if} (${\rm out}[u]<j$) \\
\> \> \> ${\rm out}[u]=j$; \\
\> \> {\bf if} ($u$ corresponds to a disjoint union) \\
\> \> \> $i={\rm in}[u]$; \\
\> \> {\bf else} \>\> $\blacktriangleright$  $u$ corresponds to an order operation \\
\> \> \> $i={\rm out}[v]+1$; \\
\> \} \\
\} \\
{\bf return} ${\rm out}[u]$;
\end{tabbing}
\normalsize
\label{fig:algorithm3x}
\caption{Computing an oriented coloring for an oriented co-graph.}  

\end{algorithm}

\begin{theorem}\label{algop}
Let $G$ be an oriented  co-graph. Then,  an optimal oriented coloring for $G$ and
$\chi_o(G)$  can be computed in linear time.
\end{theorem}

\begin{proof}
Let $G=(V,A)$ be an oriented co-graph. Using the method of \cite{CP06} we can
build a di-co-tree $T$ with root $r$ for $G$ in linear time. Further by Lemma \ref{lem-level},
we can
assume that $T$ is a canonical di-co-tree. For some node $u$ of $T$
we define by $T_u$ the subtree of $T$ which is rooted at $u$ and by
$G_u$ the subgraph of $G$ which is defined by $T_u$.
Obviously, for every vertex $u$ of $T$ the tree $T_u$ is a di-co-tree for the
digraph $G_u$ which is also an oriented co-graph.

Next, we verify that procedure {\sc Label}$(G,r,1)$, shown in Algorithm \ref{fig:algorithm3x},
returns the value $\chi_o(G)$ and computes an oriented coloring for $G$ within array
${\rm color}[u]$.
Therefore, we recursively show for every vertex $u$ of $T$ that after performing {\sc Label}$(G,u,i)$
for all leaves $u$ of $T_u$ the value ${\rm color}[u]$
leads to an oriented coloring of $G_u$ using the
colors $\{i={\rm in}[u],\ldots, {\rm out}[u]\}$  (Please note that using colors starting at values greater
than 1 is not a contradiction to Definition \ref{def-oc}.)
and the value ${\rm out}[u]-{\rm in}[u]+1$ leads to the oriented chromatic number of $G_u$.

We distinguish the following three cases depending on the type of operation
corresponding to the vertices $u$ of $T$.
\begin{itemize}
\item If $u$ is a leaf of $T$, then ${\rm color}[u]={\rm out}[u]={\rm in}[u]$
by the algorithm leads to an oriented coloring of~$G_u$.

Further, ${\rm out}[u]-{\rm in}[u]+1=1$, which obviously
corresponds to the  oriented chromatic number of $G_u$.

\item Let $u$ be an inner vertex of $T$ which corresponds to an order operation
and $u_1,\ldots,u_\ell$ are the children of $u$ in $T$.

We already know that the oriented colorings of $G_{u_i}$, $1\leq i \leq \ell$,
are feasible. Further, for $1\leq i \neq j \leq \ell$, the algorithm's way
of working ensures that a vertex from $G_{u_i}$ and a vertex
from  $G_{u_j}$ are never colored equally in $G_u$.
For $1\leq i \neq j \leq \ell$, the arcs between a
color class of $G_{u_i}$ and a color class of  $G_{u_j}$ 
have the same direction by the definition of the order operation.

By the algorithm,
value ${\rm out}[u]-{\rm in}[u]+1$
is equal to $\sum_{i=1}^\ell \chi_o(G_{u_i})$. By Lemma \ref{le-dec},
we conclude that ${\rm out}[u]-{\rm in}[u]+1$
is equal to $\chi_o(G_{i_1}\oslash \ldots \oslash G_{i_\ell})=\chi_o(G_u)$.

\item Let $u$ be an inner vertex of $T$ which corresponds to a disjoint union operation
and $u_1,\ldots,u_\ell$ are the children of $u$ in $T$.

We already know that the oriented colorings of $G_{u_i}$, $1\leq i \leq \ell$,
are feasible. Since a disjoint union operation does not create any arcs,
no two adjacent vertices have the same color in $G_u$.
Further, our method ensures that for every arc $(u,v)$ in $G$ it holds that ${\rm color}[u]
< {\rm color}[v]$. Thus, all arcs between two color classes in $G_u$ have the same
direction.

By the algorithm,  value ${\rm out}[u]-{\rm in}[u]+1$
is equal to $\max(\chi_o (G_1), \ldots, \chi_o(G_\ell))$. By Lemma \ref{le-dec},
we conclude that ${\rm out}[u]-{\rm in}[u]+1\leq \chi_o(G_1 \oplus \ldots\oplus G_\ell)=\chi_o(G_u)$.
The relation ${\rm out}[u]-{\rm in}[u]+1\geq \chi_o(G_1 \oplus \ldots\oplus G_\ell)=\chi_o(G_u)$
holds by the feasibility of our oriented coloring.

\end{itemize}
By applying the invariant for $u=r$, the statements of the theorem
follow.
\end{proof}

\begin{example}\label{ex-algo}We illustrate the method given  in Algorithm \ref{fig:algorithm3x}
by computing an oriented coloring for the oriented co-graph $G$, which is given by the
canonical di-co-tree $T$ of  Figure \ref{ex-cot}. On the left of each vertex $u$ of $T$,
the values  ${\rm in}[u]$  and ${\rm out}[u]$ are given. An optimal oriented coloring for $G$ is
given in blue letters below the leaves of $T$. The~root $r$ of $T$
leads to $\chi_o(G)={\rm out}[r]=5$.
\end{example}

\begin{figure}[ht]
\centering
{\includegraphics[width=\textwidth]{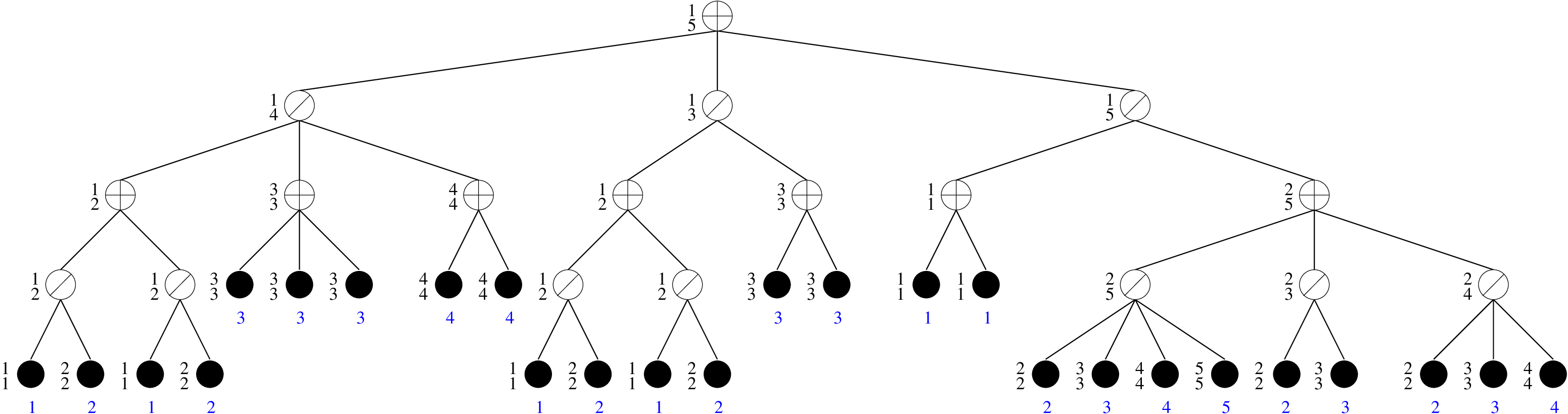}}
\vspace{3pt}
\caption{Canonical di-co-tree $T$ for oriented co-graph $G$.}\label{ex-cot}
\end{figure}

Next, we can improve the result of Lemma \ref{le-dec}~(\ref{le-dec-b})
for oriented co-graphs.

\begin{corollary}\label{le-dec2}
Let $G_1,\ldots,G_k$ be $k$ vertex-disjoint oriented co-graphs. Then, it holds that
$$\chi_o(G_1 \oplus \ldots \oplus G_k) = \max(\chi_o (G_1), \ldots,\chi_o(G_k)).$$
\end{corollary}

\begin{proof}
Let $G=G_1 \oplus \ldots \oplus G_k$ be an oriented co-graph and $T$ be a
di-co-tree  with root $r$ for $G$.
The~method given  in Algorithm \ref{fig:algorithm3x} computes
an oriented coloring using $\chi_o(G)=\chi_o(G_1 \oplus \ldots\oplus G_k)$
colors. Further, the proof of Theorem \ref{algop} shows that
$\chi_o(G_1 \oplus \ldots \oplus G_k) = \max(\chi_o (G_1), \ldots,\chi_o(G_k))$.
\end{proof}

\begin{corollary}
Let $G$ be an oriented  co-graph.  The following properties are equivalent:
\begin{enumerate}
\item $G$ is an oriented clique.
\item $G$ has a di-co-tree, which does not use any
disjoint union operation.
\item $G$ is a transitive
tournament.
\end{enumerate}
\end{corollary}

Further characterizations for transitive tournaments
and oriented co-graphs, which are oriented cliques,
can be found in \cite{Gou12}~(Chapter 9).

As mentioned in Observation
\ref{oc-tournament}, oriented coloring of an oriented graph $G$ can be characterized by the
existence of homomorphisms to tournaments.
These tournaments are not necessarily transitive and $G$ is not necessarily homomorphically equivalent to some tournament.
For oriented co-graphs we can show a deeper result.

\begin{corollary}\label{cor-hom-tt}
There is an oriented $k$-coloring of an oriented co-graph $G$
if and only if there is a homomorphism from $G$ to some transitive tournament $\overrightarrow{T_k}$ on $k$ vertices.
Further, the oriented chromatic number of an oriented co-graph $G$ is the minimum number $k$, such that
$G$ is homomorphically equivalent with the transitive tournament~$\overrightarrow{T_k}$.
\end{corollary}

\begin{proof}
Within an  oriented co-graph $G=(V,A)$
the color classes $V_1,\ldots,V_k$ of an oriented $k$-coloring define a
transitive tournament $\overrightarrow{T_k}=(\{V_1,\ldots,V_k\},\{(V_i,V_j)~|~ v_i\in V_i, v_j\in V_j, (v_i,v_j)\in A\})$.
If~$k=\chi_o(G)$, then there is a homomorphism from $\overrightarrow{T_k}$ to $G$.
\end{proof}

\section{Longest Oriented Path for Oriented Graphs}

\begin{itemize}
\item[\textbf{Name}] Oriented Path ($\OP$)

\item[\textbf{Instance}]  An oriented graph $G=(V,A)$ and a positive integer $k \leq |V|$.

\item[\textbf{Question}]  Is there an oriented path of length at least $k$ in $G$?
\end{itemize}

We can bound the path length through the oriented chromatic number,
when considering oriented co-graphs. Please note that in the subsequent results
the oriented path $\overrightarrow{P_{k+1}}$ does not necessarily refer to an induced path
(though, its definition implies no chord on the path).

\begin{proposition}[\cite{HN04}]
A directed graph $G$ has a homomorphism to the transitive tournament $\overrightarrow{T_k}$ if and only if there is no
homomorphism of the oriented path $\overrightarrow{P_{k+1}}$ to $G$.
\end{proposition}

By Corollary \ref{cor-hom-tt} this leads to the next result.

\begin{corollary}\label{cor-upper}
The oriented chromatic number of an oriented co-graph $G$ is the minimum number $k$ such that  there is no
homomorphism of the oriented path $\overrightarrow{P_{k+1}}$ to $G$.
\end{corollary}

In order to compute the length of a longest oriented
path $\ell(G)$ for an oriented graph $G$, we give the next result.

\begin{lemma}\label{le-ell}Let $G_1,\ldots,G_k$ be $k$ vertex-disjoint oriented graphs.
\begin{enumerate}

\item $\ell(G_1 \oplus \ldots\oplus G_k) = \max(\ell(G_1), \ldots, \ell(G_k))$

\item $\ell(G_1\oslash \ldots \oslash G_k)= \ell(G_1) + \ldots + \ell(G_k) + k-1$
\end{enumerate}
\end{lemma}

\begin{theorem}\label{algopath}
Let $G$ be an oriented  co-graph. Then, the length of a longest oriented
path $\ell(G)$ can be computed in linear time.
\end{theorem}

\begin{proposition}\label{prop-pathle}
Let $G$ be an oriented  co-graph. Then, it holds that  $\ell(G)=\chi_o(G)-1$.
\end{proposition}

\begin{proof}
The statement can be shown recursively on the structure of
an oriented  co-graph $G$.
\begin{itemize}
\item If $G=\bullet$, then it obviously holds $\ell(G)=0=\chi_o(G)-1$.

\item If $G=G_1 \oplus \ldots\oplus G_k$
$$
\begin{array}{lcll}
\ell(G) & =  & \ell(G_1 \oplus \ldots\oplus G_k) \\
   & =  &  \max(\ell(G_1), \ldots, \ell(G_k)) & \text{ by Lemma } \ref{le-ell} \\
   & =  &  \max(\chi_o (G_1)-1, \ldots,\chi_o(G_k)-1)                                        & \text{ by induction hypothesis}  \\
  & =  &  \max(\chi_o (G_1), \ldots,\chi_o(G_k))-1                                    \\
   & = & \chi_o(G_1 \oplus \ldots \oplus G_k)-1   & \text{ by Lemma } \ref{le-dec2}\\
    & = &\chi_o(G)-1
\end{array}
$$

\item If $G=G_1  \oslash \ldots \oslash G_k$
$$
\begin{array}{lcll}
\ell(G) & =  & \ell(G_1  \oslash \ldots \oslash G_k) \\
   & =  &  \ell(G_1) + \ldots + \ell(G_k) + k-1 & \text{ by Lemma } \ref{le-ell} \\
   & =  &   \chi_o(G_1)-1 +  \ldots +           \chi_o(G_k)-1   + k-1                            & \text{ by induction hypothesis}  \\
  & =  &   \chi_o(G_1) +  \ldots +           \chi_o(G_k)   -1                                   \\
   & = &  \chi_o(G_1\oslash \ldots \oslash G_k) -1  & \text{ by Lemma } \ref{le-dec2}\\
    & = &\chi_o(G)-1
\end{array}
$$
\end{itemize}
This shows the statements of the lemma.
\end{proof}

The previous lemma 
implies that for every oriented co-graph
the upper bound of Corollary~\ref{cor-upper} is~strict.

In order to state the next result, let $\omega(G)$
be the number of vertices in a largest clique in graph $G$.

\begin{corollary}\label{cor-last} Let $G$ be an oriented  co-graph, then
$\chi_o(G)=\ell(G)+1=\chi(\un(G))=\omega(\un(G))$ and all values can be
computed in linear time.
\end{corollary}

\begin{proof}
The first equality holds by Proposition \ref{prop-pathle} and remaining equality follows
by Lemma \ref{colo-und}.
\end{proof}

\section{Graph Isomorphism for Oriented Co-Graphs}

The isomorphism problem for undirected co-graphs has been shown to be solvable
in polynomial time in \cite{CLS81}. This result can be improved as follows.
Two undirected co-graphs $G_1$ and $G_1$ are isomorphic if and only if their canonical co-trees $T_1$
and $T_2$  are isomorphic. A canonical co-tree for a co-graph can be determined in linear time. Thus,
by applying a linear time isomorphism test for rooted labeled trees (cf. \cite{AHU74},~Section 3.2) on canonical
co-trees  for $G_1$ and $G_2$, one can decide in linear time whether $G_1$ and $G_2$ are isomorphic.

We consider the corresponding problem for oriented co-graphs.

\begin{itemize}
\item[\textbf{Name}] Oriented Co-Graph Isomorphism ($\OCI$)

\item[\textbf{Instance}]  Two oriented co-graphs $G_1=(V_1,A_1)$ and $G_2=(V_2,A_2)$.

\item[\textbf{Question}]  Are $G_1$ and $G_2$ isomorphic, i.e., is there a bijection $b:V_1 \to V_2$ such
that for all $u,v\in V_1$ it holds that $(u,v)\in A_1$ if and only if $(b(u),b(v))\in A_2$?
\end{itemize}

For oriented co-graphs, the  method using an isomorphism test
for rooted labeled trees on the co-trees does not work, since
the order of the vertices, which are representing order operations  in the di-co-tree,
must be preserved.
The procedure in Algorithm \ref{fig:algorithm4x} provides a solution for di-co-trees.

\begin{algorithm}[ht]
{\strut\footnotesize \bf procedure {\sc Test}$(T_1,T_2)$} 
\footnotesize
\smallskip
\begin{tabbing}
xxxx \= xxxx \= xxxx \= xxxx \= \kill
let $h$ be the height of $T_1$ and $T_2$  \\
{\bf for} $\ell=h$ {\bf downto} 0  {\bf do} \\
\> {\bf for all} vertices  $v$ on level  $\ell$ in $T_1$ {\bf from left to right do}   \\
\>\> {\bf if} ($v$ is a leaf) \\
\>\>\> ${\rm label}[v]=0$ \\
\>\> {\bf else}  \\
\>\>\> let $v_1,\ldots,v_r$ be the children of $v$ \\
\>\>\> ${\rm label}[v]=({\rm label}[v_1],\ldots,{\rm label}[v_r])$ \\
\>\>\> {\bf if} ($v$ corresponds to a union operation) \\
\>\>\> \> sort vector ${\rm label}[v]$ ascending\\
\> let $S_1$ be the sequence of all ${\rm label}[v]$ for all $v$ on level  $\ell$ in $T_1$ \\
\> {\bf for all}  vertices $v$ on level  $\ell$ in $T_2$ {\bf from left to right do}   \\
\>\> {\bf if} ($v$ is a leaf) \\
\>\>\> ${\rm label}[v]=0$ \\
\>\> {\bf else}  \\
\>\>\> let $v_1,\ldots,v_r$ be the children of $v$ \\
\>\>\> ${\rm label}[v]=({\rm label}[v_1],\ldots,{\rm label}[v_r])$ \\
\>\>\> {\bf if} ($v$ corresponds to a union operation) \\
\>\>\> \> sort vector ${\rm label}[v]$ ascending\\
\> let $S_2$ be the sequence of all ${\rm label}[v]$ for all $v$ on level  $\ell$ in $T_2$ \\
\> sort $S_1$ to obtain $S'_1$ and sort $S_2$ to obtain $S'_2$ \\
\> {\bf if} ($S'_1\neq S'_2$) \\
\>\> {\bf return} $false$\\
\> let $V_\ell$ be the set of all vectors on level $\ell$ in $T_1$ \\
\> find a bijection  $b: V_\ell \to \{1,\ldots, |V_\ell|\}$ \\
\> {\bf for all} vertices  $v$ on level  $\ell$ in $T_1$ {\bf do}    \\
\> \>   ${\rm label}[v]=b(v)$; \\
\> {\bf for all} vertices  $v$ on level  $\ell$ in $T_2$ {\bf do}    \\
\> \>   ${\rm label}[v]=b(v)$; \\

$\}$  \\
{\bf return} $true$;
\end{tabbing}
\caption{Testing graph isomorphism for two  oriented co-graphs given by  canonical di-co-trees.} 
\label{fig:algorithm4x}
\end{algorithm}

\begin{theorem}\label{algop2}
Let $G_1$ and $G_2$  be two oriented  co-graphs, then oriented co-graph  isomorphism
for $G_1$ and $G_2$ can be solved in linear time.
\end{theorem}

\begin{proof}
  Let $G_1$ and $G_2$ be two oriented co-graphs with the corresponding di-co-trees $T_1$ and $T_2$ , which can be found in linear time with the method of \cite{CP06}.
  Moreover we can assume, that the di-co-trees are canonical by Lemma \ref{lem-level}.
  If two graphs are isomorphic, the two canonical di-co-trees must be isomorphic, too.
  W.l.o.g. assume that the height and the roots of $T_1$ and $T_2$ are equal.
  Then, if two trees are isomorphic, there must be a bijection from the vertices of $T_1$ of level $\ell$ to the vertices of $T_2$ of level $\ell$.
  We look at the procedure from Algorithm \ref{fig:algorithm4x}. 
  Under the given conditions, the operation of the vertices of level $\ell$ are either order compositions or disjoint unions for both trees.
  If the operation on level $\ell$ is a directed union, the labels of the children of each node on level $\ell$ are sorted.
  Otherwise, it is an order composition, where the order of the children cannot be changed, such that the labels of the children will stay in the same order.
  After visiting every vertex on level $\ell$, the vectors with the labels of the children are sorted in the sequences $S_1$ and $S_2$.
  With the method given in \cite{AHU74}~(Section 3.2)~the sorting can be done in linear time with respect to the number of edges from each vertex to its children.
  If both sequences are equal, the algorithm continues, since the isomorphism is satisfied for level $\ell +1$.
  If it is not, the ordered sequences will be different, such that the algorithm terminates and returns false.
  This is repeated for every level of both trees, except for level $0$, which is the root, where the operations are assumed to be equal, and level $h$, which is the first level the algorithm goes through. When the leaves on this level are labeled, there is nothing more to do, since these vertices have no children.
  The~isomorphism of level $h$ is checked on level $h-1$.
  Let $n$ be the number of vertices in $T_1$ and $T_2$ and $m$ the number of edges. Then, the algorithm needs $2n$ steps for looking at every vertex of both trees, additional to $2m$ steps for looking at the children of each vertex. Thus, it runs in linear time.
\end{proof}

\section{Conclusions and Outlook}\label{sec-con}

In this paper, we have considered vertex coloring
on oriented graphs. We were able to introduce linear time solutions for
the oriented coloring problem, longest oriented path problem and isomorphism problem
on oriented co-graphs.
Our solutions are based on computations along a (canonical) di-co-tree for the given input co-graphs.
Furthermore, it turns out that within oriented co-graphs, the
oriented chromatic number is equal to the length of a longest oriented path plus one.
This is a quite interesting result, as within undirected co-graphs even
for bipartite graphs the path length
can not be bounded by the chromatic number.

Further, on oriented co-graphs an independent set of largest size
can be computed in the same way as known for undirected co-graphs \cite{CLS81}.
Additionally,  a tournament subdigraph of largest size and
a partition of the vertex set
into a minimum number of tournaments can be computed by applying
the method for independent set or oriented coloring problem
on the reverse input graph.

Our result concerning the  vertex coloring on oriented co-graphs use a dynamic
programming along a di-co-tree for the given input co-graph.
A similar result can be shown by using the fact that oriented co-graphs
are transitive and thus do not contain the digraph
\mbox{$\bullet \rightarrow \bullet \rightarrow\bullet\leftarrow \bullet$}
as an induced subdigraph. Using a characterization of
Chv\'{a}tal \cite{Chv84,Maf03}, we know that oriented co-graphs
are perfectly orderable graphs and therefore  any greedy coloring from a topological ordering will be optimal in both, oriented and non-oriented sense.

The given dynamic programming solutions 
provide a useful basis for exploring the complexity of $\OCN$ related to width parameters.
It remains to consider the existence of an FPT-algorithm for $\OCN$ w.r.t.\ the parameter
directed tree-width as given in \cite{JRST01}.
Since the directed tree-width of a digraph is always less or equal the undirected
tree-width of the corresponding underlying undirected graph \cite{JRST01}, the FPT-algorithm
of Ganian \cite{Gan09} (see also Section \ref{sec-intro}) does not imply such a result.

In  \cite{Gan09}, Ganian has shown that $\OCN$ is DET-hard for classes of
oriented graphs, such that the underlying undirected class has bounded rank-width.
He used a reduction from the isomorphism problem for tournaments, which has been shown
to be DET-hard in \cite{Wag07}. The same reduction also works for several
linear width parameters, since these can define the  disjoint union of two arbitrarily
large cliques. Consequently, $\OCN$ is DET-hard for classes of
oriented graphs, such that the underlying undirected class has linear NLC-width
at most $2$, linear clique-width at most $3$, neighbourhood-width at most $2$, or
linear rank-width $1$. Further, $\OCN$ is DET-hard for classes of
oriented graphs, such that the underlying undirected class has NLC-width 1 or equivalently
clique-width $2$.
The complexity of $\OCN$ on oriented graphs, such that the underlying undirected class has linear NLC-width
at most $1$ (equivalently neighbourhood-width $1$) or linear clique-width at most $2$,
remains open, since these classes do not contain the  disjoint union of two arbitrarily
large cliques.

It also remains to generalize the shown results for oriented coloring
on oriented graphs of bounded directed clique-width as given in \cite{CO00,GWY16}.
By the existence of a  monadic second order logic formula it follows that  for every $c$
the problem $\OCN_{c}$ is fixed parameter tractable
w.r.t.\ the parameter directed clique-width.

For the more general problem $\OCN$  the existence of an XP-algorithm or even an FPT-algorithm
w.r.t.\ the directed clique-width of the input graph is still open.
%
Since the directed clique-width of a digraph is always greater or equal the undirected
clique-width of the corresponding underlying undirected graph \cite{GWY16}, the result
of Ganian \cite{Gan09} (see also Section \ref{sec-intro}) does not imply a hardness result.


\end{document}